\documentclass[final,3p, times]{elsarticle}

\usepackage{algorithmic}
\usepackage{algorithm}
\usepackage{epsfig}
\usepackage{amssymb,amsmath,amsthm}
\usepackage{tabularx}

\newtheorem{theorem}{Theorem}
\newtheorem{definition}{Definition}
\newtheorem{remark}{Remark}

\newcommand{\ds}{\displaystyle}
\newcommand{\e}{\varepsilon}

\usepackage{amssymb}

\journal{Journal of Computer and System Sciences}

\begin{document}

\begin{frontmatter}



\title{Theoretical Analysis and Tuning of Decentralized Probabilistic Auto-Scaling}


\author{Bogdan Alexandru Caprarescu, Eva Kaslik, Dana Petcu}

\address{Faculty of Mathematics and Computer Science, West University of Timi\c{s}oara, Timi\c{s}oara 300223, Rom\^{a}nia\\
Research Institute e-Austria Timi\c{s}oara, Timi\c{s}oara 300223, Rom\^{a}nia}

\begin{abstract}
A major impediment towards the industrial adoption of decentralized distributed systems comes from the difficulty to theoretically prove that these systems exhibit the required behavior. In this paper, we use probability theory to analyze a decentralized auto-scaling algorithm in which each node probabilistically decides to scale in or out. We prove that, in the context of dynamic workloads, the average load of the system is maintained within a variation interval with a given probability, provided that the number of nodes and the variation interval length are higher than certain bounds. The paper also proposes numerical algorithms for approximating these minimum bounds.
\end{abstract}


\end{frontmatter}	


\section{Introduction}
\label{sec:intro}

Scalability has always been a critical non-functional requirement for parallel and distributed systems that receive a variable workload, but with the advent of cloud computing, scalability is no longer bounded by the amount of physical resources that were initially allocated to the system. The on-demand provisioning of virtualized computing nodes in the cloud allows providers of Internet services to rent only the amount of resources needed to serve the current workload and then scale in and out to cope with the request rate variations \cite{armbrust2010}.

The dynamic and unpredictable nature of the workload experienced by many Web applications, which may suddenly become popular or loose most of their users in favor of other more trendy applications, requires autonomic scaling mechanisms. Therefore, the autonomic provisioning of virtualized resources has emerged as a rich research direction leading to the proposal of many auto-scaling techniques. However, most of these techniques are executed by a central manager which, despite being able to apply advanced optimization algorithms, acts as a scalability bottleneck and introduces a single point of failure (see, for example, \cite{almeida2010, ghanbari2012, sharma2011}). The scalability limitations of centralized management are very well described by Meng et al. in \cite{meng2010} based on their experience with VMware.

P2P technologies have proved their effectiveness in building Internet applications that are both massively scalable and fault tolerant \cite{keong2005}. With the advent of large cloud data centers offering on demand access to computing resources, we argue that P2P overlay networks represent a viable solution for building elastic service systems that are capable to adapt their resource consumption to the dynamic workload. However, the shift from an Internet-based P2P environment to a cloud-based P2P environment requires a change in the way the system scales. Thus, in Internet-based P2P systems the peers join or leave the system at their will, while in a cloud environment the system itself should decide to scale in or out.

To overcome the limitations of centralized management, we proposed a decentralized probabilistic auto-scaling (DEPAS) algorithm in which each node decides to add a new node, remove itself, or do nothing in a probabilistic manner. The nodes self-organize to form an unstructured overlay network where, through gossiping, each node is able to estimate the average load of the system. Therefore, facing a variable workload, the system is capable to resize itself with the aim of keeping the average load close to a given threshold, called desired load. The difficulty of decentralized systems to maintain a fixed value for a global property was overcome by allowing the average load to vary in an interval (called load variation interval) centered in the desired load. The presentation of the DEPAS algorithm was the subject of another paper \cite{caprarescu2011} where we showed through extensive simulations of up to ten thousand nodes that the number of allocated nodes is close to the optimal one while the algorithm is highly scalable and robust. In this paper, we use probability theory to analyze the behavior of the algorithm and provide the potential customer of DEPAS with methods for tuning its parameters in such a way that a certain level of trust can be guaranteed.

Actually, one of the main challenges towards the widespread use of decentralized systems in the industry comes from the difficulty to guarantee that the system will exhibit the expected behavior. Moreover, the achievement of a global property usually depends on the proper configuration of some system parameters. Finding such a proper configuration represents the tricky part of designing a decentralized system. For the DEPAS algorithm the global property to be maintained is the average load of the system and the main configuration parameter is the length of the load variation interval. By keeping the average load within the variation interval the system avoids oscillations (i.e., additions and removals mixed in a row) and the total number of nodes stays close to the optimal value.

Due to the randomized and decentralized nature of our algorithm it is not possible to guarantee that a global property is achieved in absolute terms. Therefore, we adopt a probabilistic approach and prove that the algorithm keeps the average load within the variation interval with a certain probability, called correctness probability. Intuitively, given the length of the load variation interval, the correctness probability increases with the increase in the number of nodes. Reversely, for a certain number of nodes, the correctness probability grows with the increase of the load variation interval. Thus, in both scenarios, we are interested if there are some minimum bounds so that a minimum correctness probability can be guaranteed for any configuration in which the number of nodes and the interval length are higher than the minimum bounds, respectively.

To address the above problem, in this paper, the DEPAS algorithm is modeled as a set of Bernoulli trials. In this way, using probability theory, we formally prove that: (i) given the desired load, the load variation interval, and a minimum correctness probability, there is a minimum number of nodes so that the correctness probability is higher than the given threshold for any actual load and for any number of nodes higher than the given minimum; (ii) given the desired load, the number of nodes and a minimum correctness probability, there is a minimum interval length so that the correctness probability is higher than the minimum threshold for any actual load and for any interval length higher than the minimum. Additionally, from the formal analysis of DEPAS, we derive numerical algorithms for computing upper bounds on the minimum thresholds (based on the Chernoff-Hoeffding bounds) and compare these upper bounds with an estimation of the real minimums (based on the binomial formula for Bernoulli trials). The algorithms for both scenarios can be used to tune DEPAS at design time. Moreover, the algorithm corresponding to the second scenario can be also used at runtime to dynamically adjust the variation interval length of each node according to the continuously changing size of the system.

The remaining of this paper is organized as follows. Section \ref{sec:pas} describes the DEPAS algorithm and formulates the problem to be solved. The problem is formally analyzed and a proof is given in Section \ref{sec:analysis}.  The algorithms for tuning DEPAS are described in Section \ref{sec:alg}, while Section \ref{sec:tests} discusses the tuning algorithms in the light of some experimental results. Related work is shown in Section \ref{sec:related}. Finally, Section \ref{sec:conclusion} concludes the paper. 

\section{Decentralized probabilistic auto-scaling (DEPAS)}
\label{sec:pas}

The DEPAS algorithm was described in \cite{caprarescu2011} as part of a middleware for deploying massively scalable services in a cloud infrastructure. In this section, the algorithm is introduced with an emphasis on the problem of finding a subset of the parameter configurations for which a minimum correctness probability can be guaranteed.

We assume a system composed of a set of computing nodes. A node can be either a physical machine or a virtual machine allocated from a cloud provider on a utility basis. The nodes are homogenous in the sense that they have the same hardware configuration and run the same piece of software comprising both the functional service and the components of our middleware: overlay management, load balancing, and auto-scaling. The overlay management algorithm organizes the nodes into an unstructured overlay network where each node has a fixed degree and a low standard deviation in-degree. This is an adapted version of the gossip-based protocol developed by Jelasity et al. \cite{jelasity2007}. For load balancing we use a combined approach: a DNS is employed to assign the address of a node to each client while an internal decentralized load balancing protocol (such as the dimension exchange protocol \cite{dinitto2008}) moves requests between neighbors in order to equalize the load across nodes.

In the following, the term node will be used to denote both the machine and the service instance running on that machine. Thus, the main parameters of a node are the capacity and the load. The capacity of a node is the maximum number of requests per second that can be processed by the service deployed on that node and is derived through benchmarking. As the nodes are homogenous, they have the same capacity. The load of a node is computed at a given moment in time, as the percentage of the average number of requests per second that were scheduled on that node over a certain timeframe with respect to the capacity of the node. Then, in a homogenous system, the average load of the system is simply the arithmetic average of the loads of all nodes. Note that in the case when the workload received by the system overcomes its capacity, the average load is supra-unitary.

\begin{table}
	 \caption{DEPAS notations}
\begin{center}
\begin{tabularx}{\textwidth}{ | p{0.7cm} | X | }
	\hline
	$n$ & Number of nodes of the system \\
	$L_0$ & Desired load threshold  (percent with respect to node capacity) \\
	$L$ & Average load of the system (percent with respect to node capacity) \\
	$\delta$ & Defines the allowed load variation \\
	$pi$ & Probability indicator: used to compute the node probability, $p$ \\
	$p$ & Node-level probability used to make node addition/removal decisions\\
	$P_0$ & Minimum correctness probability threshold \\
	\hline
\end{tabularx}
\end{center}
   
    \label{table:notations}
\end{table}

Under these considerations, the goal of the auto-scaling algorithm is to maintain the average load of the system, noted with $L$, within a given interval, $\left(L_0 - \delta, L_0 + \delta \right)$, where $L_0$ is called desired load and $\delta$ defines the allowed variation of the load (see Table \ref{table:notations} for a complete list of notations). To achieve this goal each node can execute two types of actions: remove itself and allocate one or more other nodes.

The DEPAS algorithm is shown in listing \ref{alg:pas}. It is periodically run by each node and begins by retrieving an estimation of the average load of the system. Note that the average load is not computed at this time, but just retrieved from the component running the average protocol. Although there are gossip-based solutions for computing the average load of the system at each node \cite{jelasity2005}, in our experiments \cite{caprarescu2011}, for simplicity and higher scalability, we approximated the average load of the system by the average load of the node and its neighbors. If the load is less than or equal to $L_0 - \delta$, then the node computes a probability indicator using Eq. (\ref{eq:prob-indicator}) and, because the indicator is sub-unitary in this case, the node uses it as the probability to remove itself. Otherwise, if the load is higher than or equal to $L_0 + \delta$, then the probability indicator is computed using the same equation. In this situation, the indicator can be supra-unitary where its integer part represents the number of nodes to be added for sure, while its fractional part is used as the probability to add another node. Note that the \textit{random()} function generates a uniformly distributed random decimal number between 0 and 1.

\begin{algorithm}
\caption{DEPAS}
\label{alg:pas}
\begin{algorithmic}
\WHILE{$true$}
\STATE  $wait(timeframe)$
\STATE $L \gets getAverageSystemLoad()$
\IF {$L \le L_0 - \delta$}	
       \STATE $pi \gets computeProbabilityIndicator(L, L_0)$
       \STATE $p \gets pi$
	 \IF {$p < random()$}
        	\STATE $removeSelf()$
	\ENDIF
\ELSE
	\IF {$L \ge L_0 + \delta$}
      	\STATE $pi \gets computeProbabilityIndicator(L, L_0)$
		\STATE $m \gets \lfloor pi \rfloor$
		\STATE $p \gets \left\{pi\right\}$
		\IF {$p < random()$}
        		\STATE $m \gets m + 1$
		\ENDIF
		\STATE $addNodes(m)$
	\ENDIF
\ENDIF
\ENDWHILE
\end{algorithmic}
\end{algorithm}

\begin{equation}
\label{eq:prob-indicator}
pi =  \frac{\left| L - L_0 \right|}{L_0}
\end{equation}

The desired load, $L_0$, is subject of a tradeoff. On one hand, a high value reduces the number of nodes, but, in case of a sudden workload increase, leads to a severe degradation of the system performance before the system has the chance to allocate new nodes. On the other hand, a low desired load increases the tolerance of the system to sudden traffic bursts at the cost of allocating many nodes. As the desired load has a high impact on both performance and cost, the customer of DEPAS is in the best position to set its value.

However, while the meaning and impact of the desired load is straightforward for the customer, this is not the case when it comes to the load variation threshold, $\delta$. Actually, even though the customer should not need to care about $\delta$, they are definitely interested in two aspects of the algorithm, namely correctness and accuracy, that are directly impacted by $\delta$. By correctness we understand the ability of the system to make a right provisioning decision in the first place, thus avoiding oscillations (i.e., additions and removals mixed in a row). Due to the randomized nature of DEPAS it is not possible to evaluate its correctness in absolute terms (i.e., correct or incorrect). Therefore, we introduce the notion of correctness probability and allow the customer to define the correctness of the algorithm by specifying a minimum threshold of the correctness probability (see Definition \ref{def:correctness-probability}). Then, the accuracy of a correct algorithm is a measure of how close the allocated number of nodes is to the optimal one. Therefore, the accuracy of DEPAS is directly influenced by $\delta$: the lower $\delta$ is, the higher the accuracy is. In this paper, we provide a formal analysis of the link between correctness probability, load variation threshold, and number of nodes. Our goal is to provide the customer with an analytical method for configuring DEPAS so that it meets given correctness requirements with an as high as possible accuracy.

\begin{definition}
\label{def:correctness-probability}
Let $L_0$ and $\delta$ be the desired load and the load variation threshold of a DEPAS instance, respectively. Then, the correctness probability is the probability that, after DEPAS has been run by each node, the new average load is in the interval $\left(L_0 - \delta, L_0 + \delta\right)$. The correctness probability is noted with $P(L_0 - \delta < L < L_0 + \delta)$.
Consequently, we say that an instance of DEPAS is correct if its correctness probability is higher than or equal to a given threshold, denoted by $P_0$.
\end{definition}

More concretely, we consider two usage scenarios of the DEPAS algorithm. In both scenarios, the customer specifies the minimum correctness probability, $P_0$. In the first scenario -- called \textit{Min n} -- the customer is able to predict the minimum workload of the system and implicitly the minimum number of nodes and they are interested in the values of $\delta$ for which the correctness probability is higher than $P_0$. The second scenario -- called \textit{Min $\delta$} -- is concerned with the case when the customer has set a small value for $\delta$ in order to obtain a very accurate algorithm and wants to find out for which system sizes a correctness probability higher than $P_0$ can be guaranteed. Both scenarios are applied at design time before deploying DEPAS in the production environment. However, we can imagine the \textit{Min $\delta$} scenario being also used at runtime as a subroutine of DEPAS in order to dynamically adjust the value of $\delta$ at each node, provided that an estimation of the system size is available at each node \cite{jelasity2005, montresor2009}. By applying the \textit{Min $\delta$} scenario at runtime we expect to increase the accuracy of DEPAS. However, as opposed to the design-time version, a runtime \textit{Min $\delta$} algorithm, being executed at each node, must be fast and have a low resource consumption in order not to overload the system.

In this paper, we address the above challenges for the simplified case when each node computes the same probability indicator, $pi$. This implies that each node precisely estimates the average system load. We also assume that all nodes have synchronized clocks and simultaneously run the DEPAS algorithm, although this constraint is not needed in practice. In other words, the algorithm is assumed to work in cycles where, in each cycle, each node uses the same probability to decide upon the execution of a scaling action.

Looking back at listing \ref{alg:pas}, we notice that if $pi \geq 1$, which happens when $L \geq 2L_0$, then $n \lfloor pi \rfloor$ is the number of nodes that are added for sure (where $n$ is the number of existing nodes), which represents the deterministic part of the decision. Therefore, they have no impact on the correctness probability and the case when $L \geq 2L_0$ is reducible to the case when $L \in \left[L_0 + \delta, 2 \cdot L_0\right)$. On the other hand, the addition and removal cases are symmetric. Consequently, for simplicity and without loss of generality, we will consider only the addition case. Under these considerations, the formula for computing the probability of each node, $p$, is given by Eq. (\ref{eq:prob}).
\begin{equation}
\label{eq:prob}
p =  \frac{L - L_0}{L_0},\quad L \in [L_0 + \delta, 2L_0)~.
\end{equation}

We want to express the correctness probability in function of the number of added nodes. For this purpose, Theorem \ref{th:optimal-nb-of-nodes} defines the notion of optimal number of nodes to be added as a rational number. In Theorem \ref{th:correctness-probability-equality}, the correctness probability is expressed in function of the number of added nodes and the optimal number of nodes corresponding to the bounds of the load variation interval.

\begin{theorem}
\label{th:optimal-nb-of-nodes}
Let $L_0 \in (0, 1)$, a system composed of $n$ homogenous nodes, and $L \in [L_0 + \delta, 2L_0)$ the average load of the system. Then, the optimal number of nodes to be added to the system so that the new average load will be equal to $L_0$, denoted by $M(n, L, L_0)$, is given by
\begin{equation}
\label{eq:m}
M(n, L, L_0) = n\frac{L - L_0}{L_0}~.
\end{equation}
\end{theorem}

\begin{proof}
Let $m$ be the number of nodes to be added and $C$ the common capacity of all nodes. As the system has the same workload before and after provisioning the new nodes, we have $nLC = (n + m)L_0C$, from where it turns out that $m = n\frac{L - L_0}{L}$.
\end{proof}

\begin{theorem}
\label{th:correctness-probability-equality}
Let $L_0 \in (0, 1)$, a system composed of $n$ homogenous nodes, $L \in [L_0 + \delta, 2L_0)$ the average load of the system, and $L' \in (0, 2L_0)$ the average load of the system after the addition of $m$ nodes. Then, the correctness probability is equal to the probability of $m \in \left(M(n, L, L_0 + \delta), M(n, L, L_0 - \delta)\right)$, as expressed by the formula below:
$$P(L_0 - \delta < L' < L_0 + \delta) = P(M(n, L, L_0 + \delta) < m < M(n, L, L_0 - \delta))~.$$
\end{theorem}

\begin{proof}
As the workload of the system remains unchanged, we have $nL = (n + m)L'$, which implies that $L' = \frac{nL}{n + m}$.

Then, $P(L_0 - \delta < L' < L_0 + \delta) = P(L_0 - \delta <  \frac{nL}{n + m} < L_0 + \delta) =  P(\frac{1}{L_0 - \delta} >  \frac{n + m}{nL} > \frac{1}{L_0 + \delta}) = P(\frac{nL}{L_0 + \delta} - n < m < \frac{nL}{L_0 - \delta} - n) = P(n \frac{L - L_0 - \delta}{L_0 + \delta} < m < n \frac{L - L_0 + \delta}{L_0 - \delta}) = P(M(n, L, L_0 + \delta) < m < M(n, L, L_0 - \delta))$ (by applying Theorem \ref{th:optimal-nb-of-nodes}).
\end{proof}

In this section, we described the DEPAS algorithm and the problem we want to solve: finding whether and in which conditions a minimum correctness probability can be guaranteed. This section also prepared the ground for the formalization and theoretical analysis of the problem, which falls within the scope of the next section. 

\section{Theoretical analysis}
\label{sec:analysis}

As stated in the previous section, the DEPAS algorithm is assumed to work in cycles. In each cycle, each node uses the same probability to decide whether to add a new node or not. Consequently, a cycle of the algorithm can be modeled as a set of $n$ Bernoulli trials, $X_1, X_2, ..., X_n$. We denote by $S_n$ the outcome of the experiment, which is the number of nodes added in the respective cycle. Under this formalization, Theorem \ref{th:expected-nb-of-nodes} proves that the formula for $p$ described by Eq. (\ref{eq:prob}) was correctly chosen.  

\begin{theorem}
\label{th:expected-nb-of-nodes}
Let us consider a cycle of the DEPAS algorithm where each node adds a new node with probability $p = (L - L_0)/L_0$. Then, the expected number of added nodes is equal to the optimal number of nodes to be added.
\end{theorem}
\begin{proof}
If $S_n$ is the number of added nodes, then taking into consideration that $S_n$ has a binomial distribution of parameters $n$ and $p$ it follows that $E(S_n) = np = M(n, L, L_0)$ (according to Theorem \ref{th:optimal-nb-of-nodes}).
\end{proof}

Moreover, the correctness probability can be computed by using the binomial probability distribution:
\begin{equation}
\label{eq:binomial}
P(M(n, L, L_0 + \delta) < m < M(n, L, L_0 - \delta)) = \sum_{i=\left\lceil M(n, L, L_0 + \delta)\right\rceil}^{\left\lfloor M(n, L, L_0 - \delta) \right\rfloor} \binom{i}{n}p^i(1-p)^{n-i}
\end{equation}

The problem with the expression from the left-hand side of Eq. (\ref{eq:binomial}) is that it is discontinuous and non-monotonic with respect to $L$ and non-monotonic with respect to $n$ and $\delta$, thus increasing the risk of errors when using it with numerical optimization algorithms. Moreover, its evaluation is expensive due to the binomial coefficients. Therefore, a faster and less error-prone method is needed for estimating the correctness probability, or at least its lower bound.

Theorem \ref{th:minimum} uses the Chernoff-Hoeffding bounds \cite{mitzenmacher2005} to compute a lower bound of the correctness probability. More precisely, it proves that for any probability threshold $P_0 \in (0, 1)$, there is a minimum $\delta^\star$ (for a fixed $n$) or a minimum $n^\star$ (for a fixed $\delta$) so that the correctness probability is higher than or equal to $P_0$ for all $\delta \ge \delta^\star$ or for all $n \ge n^\star$, respectively. The proof of the theorem also provides a method for computing $\delta^\star$ and $n^\star$.

\begin{theorem}
\label{th:minimum}
Let us consider the DEPAS algorithm under the above assumptions. A cycle of the algorithm is modeled as a set of $n$ Bernoulli trials. Then, the following two affirmations hold.
\begin{itemize}
  \item[a.] For any $P_0\in(0,1)$ and $\delta=\delta_0$, there exists $n^\star=n^\star(P_0,\delta_0)\in\mathbb{Z}_+$ such that
\begin{align}\label{ineq.a}
P\left(M(n,L,L_0+\delta_0)<S_n<M(n,L,L_0-\delta_0)\right)\geq P_0,\qquad\forall~L\in(L_0+\delta,2L_0),~\forall~n\geq n^\star.
\end{align}
  \item[b.]  For any $P_0\in(0,1)$ and $n=n_0\in\mathbb{Z}_+$ satisfying
\begin{equation}\label{ineq.P0.n0}
1-\left(\frac{1}{3}\right)^{n_0}-\left(\frac{2}{3}\right)^{n_0}> P_0,
\end{equation}
there exists $\delta^\star=\delta^\star(P_0,n_0)\in(0,1)$ such that
\begin{align}\label{ineq.b}
P\left(M(n_0,L,L_0+\delta)<S_{n_0}<M(n_0,L,L_0-\delta)\right)\geq P_0,\qquad\forall~L\in(L_0+\delta,2L_0),~\forall~\delta\geq \delta^\star.
\end{align}
\end{itemize}
\end{theorem}

\begin{proof}
To simplify this problem, we will use the notations (re-scaling) $$\ds\frac{L-L_0}{L_0}=p\in(0,1)\quad\textrm{and}\quad\ds\frac{\delta}{L_0}=\e\in(0,1).$$
Denoting
\begin{align*}
Prob(p,\e,n)&=P\left(M(n,L,L_0+\delta)<S_{n}<M(n,L,L_0-\delta)\right)\\
&=P\left(\frac{L-L_0-\delta}{L_0+\delta}<\frac{S_n}{n}<\frac{L-L_0+\delta}{L_0-\delta}\right)\\
&=P\left(\frac{p-\e}{1+\e}<\frac{S_n}{n}<\frac{p+\e}{1-\e}\right),
\end{align*}
our problem is twofold:
\begin{itemize}
  \item[\textbf{a.}] For a fixed $\e=\e_0$, estimate the smallest value of $n\in\mathbb{Z}_+$ such that $Prob(p,\e_0,n)\geq P_0$, for any $p\geq \e_0$;
  \item[\textbf{b.}] For a fixed $n=n_0$, estimate the smallest value of $\ds\e\in\left(0,1\right)$ such that $Prob(p,\e,n_0)\geq P_0$, for any $p\geq \e$.
\end{itemize}

It can be easily seen that
$$Prob(p,\e,n)=1-P\left(\frac{S_n}{n}\leq \frac{p-\e}{1+\e}\right)-P\left(\frac{S_n}{n}\geq \frac{p+\e}{1-\e}\right).$$

First, as $\ds 0\leq\frac{p-\e}{1+\e}\leq p$, the theorem of Chernoff-Hoeffding bounds provides the following inequality:
$$P\left(\frac{S_n}{n}\leq \frac{p-\e}{1+\e}\right)\leq e^{-nD\left[\frac{p-\e}{1+\e},p\right]},$$
where
$$D[x,y]=x\ln\frac{x}{y}+(1-x)\ln\frac{1-x}{1-y}$$
represents the Kullback-Leibler divergence between Bernoulli distributed random variables with parameters $x$ and $y$ respectively, $x,y\in(0,1)$.

On the other hand, we notice that if $p>1-2\e$, we have $\ds\frac{p+\e}{1-\e}>1$ and hence, $\ds P\left(\frac{S_n}{n}\geq \frac{p+\e}{1-\e}\right)=0$. However, if $p\leq 1-2\e$, we have $\ds p\leq\frac{p+\e}{1-\e}\leq 1$ and  the theorem of Chernoff-Hoeffding bounds provides
$$P\left(\frac{S_n}{n}\geq \frac{p+\e}{1-\e}\right)\leq e^{-nD\left[\frac{p+\e}{1-\e},p\right]}.$$

Therefore, we obtain the following lower bounds for the probability $Prob(p,\e,n)$:
\begin{equation}\label{eq.case1}
Prob(p,\e,n)\geq 1-e^{-nD\left[\frac{p-\e}{1+\e},p\right]}-e^{-nD\left[\frac{p+\e}{1-\e},p\right]},
\quad\forall~p\in[\e,1-2\e],~\e\in\left(0,\frac{1}{3}\right),~n\in\mathbb{Z}_+
\end{equation}
and
\begin{equation}\label{eq.case2}
Prob(p,\e,n)\geq 1-e^{-nD\left[\frac{p-\e}{1+\e},p\right]},\quad \forall~p\geq\max\{\e,1-2\e\},~\e\in(0,1),~n\in\mathbb{Z}_+~.
\end{equation}

In the following, let $B_1(p,\e,n)$ be the function from the right hand side of the inequality (\ref{eq.case1})
$$B_1(p,\e,n)=1-e^{-nD\left[\frac{p-\e}{1+\e},p\right]}-e^{-nD\left[\frac{p+\e}{1-\e},p\right]},
\quad\forall~p\in[\e,1-2\e],~\e\in\left(0,\frac{1}{3}\right),~n\in\mathbb{Z}_+$$
and $B_2(p,\e,n)$ the function from the right hand side of the inequality (\ref{eq.case2})
$$
B_2(p,\e,n)=1-e^{-nD\left[\frac{p-\e}{1+\e},p\right]},\quad \forall~p\geq\max\{\e,1-2\e\},~\e\in(0,1),~n\in\mathbb{Z}_+
$$

Both functions $B_1$ and $B_2$ are continuously differentiable on their respective domains of definition. It is easy to see that
$$\frac{\partial B_1}{\partial n}(p,\e,n)=D\left[\frac{p-\e}{1+\e},p\right]e^{-nD\left[\frac{p-\e}{1+\e},p\right]}+D\left[\frac{p+\e}{1-\e},p\right]e^{-nD\left[\frac{p+\e}{1-\e},p\right]}> 0$$
because of the positiveness of the Kullback-Leibler divergence, and hence, the function $B_1$ is strictly increasing with respect to $n$.

On the other hand,
$$\frac{\partial B_1}{\partial \e}(p,\e,n)=
n\frac{\partial}{\partial \e}\left(D\left[\frac{p-\e}{1+\e},p\right]\right)e^{-nD\left[\frac{p-\e}{1+\e},p\right]}
+n\frac{\partial}{\partial \e}\left(D\left[\frac{p+\e}{1-\e},p\right]\right)e^{-nD\left[\frac{p+\e}{1-\e},p\right]}~,$$
where
$$\frac{\partial}{\partial \e}\left(D\left[\frac{p-\e}{1+\e},p\right]\right)=\frac{1-p}{(1+\e)^2}\left(\ln\left(1+\frac{2\e}{1-p}\right)-\ln\left(1-\frac{\e}{p}\right)\right)> 0$$
$$\frac{\partial}{\partial \e}\left(D\left[\frac{p+\e}{1-\e},p\right]\right)=\frac{1+p}{(1-\e)^2}\left(\ln\left(1+\frac{\e}{p}\right)-\ln\left(1-\frac{2\e}{1-p}\right)\right)> 0$$
for any $\ds\e\in\left(0,\frac{1}{3}\right)$ and $p\in[\e,1-2\e]$. Therefore, the function $B_1$ is strictly increasing with respect to $\e$ as well.

In a similar way, it follows that $B_2$ is also strictly increasing with respect to the variables $\e$ and $n$.


\medskip

\noindent\textbf{Case a.} When $\ds\e=\e_0$ is fixed (i.e. $\delta=\delta_0=\e_0L_0$), we have two subcases.

\noindent\textbf{a.1.} If $\ds\e_0\in\left(0,\frac{1}{3}\right)$ it can be proved that for any $p\in[\e_0,1-2\e_0]$, there exists a unique $n_1(p)\in(0,\infty)$ satisfying the equation
$$B_1(p,\e_0,n_1(p))=P_0.$$
Indeed, for $p\in[\e_0,1-2\e_0]$ arbitrarily fixed, the function $n\mapsto B_1(p,\e_0,n)$ is continuous, strictly increasing, $B_1(p,\e_0,0)=-1$ and $\lim\limits_{n\rightarrow\infty}B_1(p,\e_0,n)=1$. Therefore, $n\mapsto B_1(p,\e_0,n)$ is a bijective mapping between $(0,\infty)$ and the interval $(-1,1)$. As $P_0\in(0,1)$, there exists a unique solution $n_1(p)\in(0,\infty)$ of the equation $B_1(p,\e_0,n)=P_0$.
Moreover, the implicit function theorem guarantees that the function $n_1(p)$ defined above is continuously differentiable on $[\e_0,1-2\e_0]$, and hence, it is bounded. We denote
$$\ds n_1^\star=\sup\limits_{p\in[\e_0,1-2\e_0]}n_1(p).$$

On the other hand, for any $p\in(1-2\e_0,1)$, there exists a unique $n_2(p)\in(0,\infty)$ satisfying the equation
$$B_2(p,\e_0,n_2(p))=P_0,$$
given explicitly by
\begin{equation}\label{eq.n2}
n_2(p)=\frac{-\ln(1-P_0)}{D\left[\frac{p-\e_0}{1+\e_0},p\right]}.
\end{equation}
As $\ds\lim\limits_{p\rightarrow 1}D\left[\frac{p-\e_0}{1+\e_0},p\right]=\infty$, we obtain that $n_2(p)$ is bounded on the interval $(1-2\e_0,1)$  and let
$$\ds n_2^\star=\sup\limits_{p\in(1-2\e_0,1)}n_2(p).$$

Considering $$n^\star=n^\star(P_0,\e_0)=\left\lceil\max\{n_1^\star,n_2^\star\}\right\rceil\in\mathbb{Z}_+$$
from the fact that $B_1$ is increasing with respect to $n$, we clearly have
$$B_1(p,\e_0,n)\geq B_1(p,\e_0,n^\star)\geq B_1(p,\e_0,n_1^\star)\geq B_1(p,e_0,n_1(p))=P_0,\quad\forall~p\in[\e_0,1-2\e_0],~\forall n\geq n^\star,$$
and therefore,
$$Prob(p,\e_0,n)\geq P_0,\quad\forall~p\in[\e_0,1-2\e_0],~\forall n\geq n^\star.$$
Similarly, as $B_2$ is increasing with respect to $n$, we obtain
$$B_2(p,\e_0,n)\geq B_2(p,\e_0,n^\star)\geq B_2(p,\e_0,n_2^\star)\geq B_2(p,e_0,n_2(p))=P_0,\quad\forall~p\in(1-2\e_0,1),~\forall n\geq n^\star,$$
and therefore,
$$Prob(p,\e_0,n)\geq P_0,\quad\forall~p\in(1-2\e_0,1),~\forall n\geq n^\star.$$
In conclusion, we obtain:
$$Prob(p,\e_0,n)\geq P_0,\quad\forall~p\geq \e_0,~\forall n\geq n^\star.$$

\medskip

\noindent\textbf{a.2.} If $\ds\e_0\geq \frac{1}{3}$, we know from inequality (\ref{eq.case2}) that
$$Prob(p,\e_0,n)\geq B_2(p,\e_0,n),\quad\forall~p\geq \e_0,~n\in\mathbb{Z}_+.$$
For any $p\in[\e_0,1)$, the unique solution of the equation $B_2(p,e_0,n)=P_0$ is $n_2(p)$ given by (\ref{eq.n2}). Denoting
$$n^\star=n^\star(P_0,\e_0)=\left\lceil\sup\limits_{p\in[\e_0,1)}n_2(p)\right\rceil,$$
as $B_2$ is increasing with respect to $n$, we obtain
$$B_2(p,\e_0,n)\geq B_2(p,\e_0,n^\star)\geq B_2(p,\e_0,n_2(p))=P_0,\quad\forall~p\in[\e_0,1),~\forall n\geq n^\star,$$
and hence
$$Prob(p,\e_0,n)\geq P_0,\quad\forall~p\geq \e_0,~\forall n\geq n^\star.$$

\medskip

\noindent\textbf{Case b.} When $n=n_0\in\mathbb{Z}_+$ is fixed, we will consider two scenarios.

\noindent\textbf{b.1} For $\ds p\leq\frac{1}{3}$ the lower bound for the probability $Prob(p,\e,n_0)$ is given by $B_1(p,\e,n_0)$ according to (\ref{eq.case1}).

Let $p\in\left(0,\frac{1}{3}\right]$ be arbitrarily fixed. The equation
$$B_1(p,\e,n_0)=P_0$$
has at most one solution $\e(p)\leq p$. Indeed, we have $B_1(p,0,n_0)=-1$ and
$$\lim_{\e\rightarrow p}B_1(p,\e,n_0)=1-(1-p)^{n_0}-e^{-n_0D\left[\frac{2p}{1-p},p\right]}$$
It is a simple calculus exercise to show that the function
$$h_1(p)=1-(1-p)^{n_0}-e^{-n_0D\left[\frac{2p}{1-p},p\right]}$$
is strictly increasing on the interval $\left(0,\frac{1}{3}\right)$, $\displaystyle\lim_{p\rightarrow 0}h_1(p)=-1$ and
$$\lim_{p\rightarrow \frac{1}{3}}h_1(p)=1-\left(\frac{1}{3}\right)^{n_0}-\left(\frac{2}{3}\right)^{n_0}>P_0$$
according to inequality (\ref{ineq.P0.n0}). Tt follows that
there exists a unique $p_1^\star\in\left(0,\frac{1}{3}\right)$ such that $h_1(p_1^\star)=P_0$. Therefore, if $p\in(0,p_1^\star]$, the equation $B_1(p,\e,n_0)=P_0$ has no solution, but if $p\in\left(p_1^\star,\frac{1}{3}\right]$, the equation $B_1(p,\e,n_0)=P_0$ has a unique solution $\e_1(p)$, satisfying $\e_1(p)\leq p$. Moreover, from the implicit function theorem we obtain that $\e_1(p)$ is continuously differentiable on $\left(p_1^\star,\frac{1}{3}\right]$, and therefore, it is bounded.


Note that $\lim\limits_{p\rightarrow p_1^\star}\e_1(p)=p_1^\star$ because we have $\lim\limits_{\e\rightarrow p_1^\star}B_1(p_1^\star,\e,n_0)=h_1(p_1^\star)=P_0$, and hence, we can extend the function $\e_1$ by continuity, considering $\e_1(p_1^\star)=p_1^\star$.

Denoting
$$\e_1^\star=\sup\limits_{p\in\left[p_1^\star,\frac{1}{3}\right]}\e_1(p),$$
we first observe that $\e_1^\star\geq \e_1(p_1^\star)=p_1^\star$ and therefore, using the fact that $B_1$ is increasing with respect to $\e$, we obtain
\begin{equation}\label{ineq.b1}
B_1(p,\e_1^\star,n_0)\geq B_1(p,\e_1(p),n_0)=P_0,\quad\forall~p\in \left[\e_1^\star,\frac{1}{3}\right).
\end{equation}

\noindent\textbf{b.2} For $\ds p>\frac{1}{3}$, the lower bound for $Prob(p,\e,n_0)$ can be expressed from  (\ref{eq.case1}) and  (\ref{eq.case2}) as
$$B(p,\e,n_0)=\left\{
                \begin{array}{ll}
                  B_1(p,\e,n_0), & \textrm{if }\e\in\left(0,\frac{1-p}{2}\right) \\
                  B_2(p,\e,n_0), & \textrm{if }\e\in\left[\frac{1-p}{2},p\right)
                \end{array}
              \right.
$$
The function $B$ is increasing with respect to $\e$
$$\lim_{\e\rightarrow 0}B(p,\e,n_0)=\lim_{\e\rightarrow 0}B_1(p,\e,n_0)=-1$$
and
$$\lim_{\e\rightarrow p}B(p,\e,n_0)=\lim_{\e\rightarrow p}B_2(p,\e,n_0)=1-(1-p)^{n_0}>1-\left(\frac{2}{3}\right)^{n_0}>P_0$$
The function $B$ has a jump discontinuity at $\e=\frac{1-p}{2}$:
$$\lim_{\e\uparrow \frac{1-p}{2}}B(p,\e,n_0)=\lim_{\e\rightarrow \frac{1-p}{2}}B_1(p,\e,n_0)=1-e^{-n_0D\left[\frac{3p-1}{3-p},p\right]}-p^{n_0}=h_2(p),$$
$$\lim_{\e\downarrow \frac{1-p}{2}}B(p,\e,n_0)=\lim_{\e\rightarrow \frac{1-p}{2}}B_2(p,\e,n_0)=1-e^{-n_0D\left[\frac{3p-1}{3-p},p\right]}=h_3(p).$$
It is easy to check that the functions $h_2(p)$ and $h_3(p)$ defined above are continuous and strictly decreasing on the interval $\left(\frac{1}{3},1\right)$ and
$h_2(p)<h_3(p)$ for any $p\in\left(\frac{1}{3},1\right)$. Moreover:
$$\lim_{p\rightarrow\frac{1}{3}}h_2(p)=1-\left(\frac{2}{3}\right)^{n_0}-\left(\frac{1}{3}\right)^{n_0}>P_0\quad\textrm{and}\quad \lim_{p\rightarrow 1}h_2(p)=-1,$$
and
$$\lim_{p\rightarrow\frac{1}{3}}h_3(p)=1-\left(\frac{2}{3}\right)^{n_0}>P_0\quad\textrm{and}\quad \lim_{p\rightarrow 1}h_3(p)=0.$$
Therefore, there exist unique values $p_2^\star,p_3^\star\in \left(\frac{1}{3},1\right)$, $p_2^\star<p_3^\star$, such that $h_2(p_2^\star)=P_0$ and $h_3(p_3^\star)=P_0$.

If $p\in\left(\frac{1}{3},p_2^\star\right)$, it follows that
$$\lim_{\e\uparrow \frac{1-p}{2}}B(p,\e,n_0)=\lim_{\e\rightarrow \frac{1-p}{2}}B_1(p,\e,n_0)=h_2(p)>h_2(p_2^\star)=P_0$$
and hence, since $B_1$ is increasing with respect to $\e$, the equation $B_1(p,\e,n_0)=P_0$ has a unique solution $\e_2(p)$, satisfying $\e_2(p)\leq \frac{1-p}{2}<\frac{1}{3}$. Moreover, from the implicit function theorem we obtain that $\e_2(p)$ is continuously differentiable on $\left(\frac{1}{3},p_2^\star\right)$, and therefore, it is bounded. We can extended $\e_2$ by continuity by considering
$\e_2(p_2^\star)=\frac{1-p_2^\star}{2}.$

Denoting
$$\e_2^\star=\sup\limits_{p\in\left(\frac{1}{3},p_2^\star\right]}\e_2(p),$$
and using the fact that $B$ is increasing with respect to $\e$, we obtain
\begin{equation}\label{ineq.b21}
B(p,\e_2^\star,n_0)\geq B(p,\e_2(p),n_0)=B_1(p,\e_2(p),n_0)=P_0,\quad\forall~p\in\left(\frac{1}{3},p_2^\star\right].
\end{equation}

If $p\in(p_2^\star,p_3^\star]$ we have $h_3(p)\geq h_3(p_3^\star)=P_0$ and $\e_2^\star\geq \e_2(p_2^\star)=\frac{1-p_2^\star}{2}>\frac{1-p}{2}$.  Taking into consideration that $B_2$ is increasing with respect to $\e$, we obtain:
\begin{equation}\label{ineq.b22}
B(p,\e_2^\star,n_0)=B_2(p,\e_2^\star,n_0)\geq \lim_{\e\rightarrow \frac{1-p}{2}}B_2(p,\e,n_0)=h_3(p)\geq P_0,~\quad\forall~p\in(p_2^\star,p_3^\star].
\end{equation}

If $p\in(p_3^\star,1)$, since $h_3(p)<h_3(p_3^\star)=P_0$, we have that
$$\lim_{\e\rightarrow \frac{1-p}{2}}B_2(p,\e,n_0)=h_3(p)<P_0\quad\textrm{and}\quad \lim_{\e\rightarrow p}B_2(p,\e,n_0)>P_0$$
and therefore, the equation $B_2(p,\e,n_0)=P_0$ has a unique solution $\e_3(p)$ such that $\frac{1-p}{2}<e_3(p)<p$. The function $\e_3(p)$ is continuously differentiable, bounded, and it can be extended by continuity, considering $\e_3(p_3^\star)=\frac{1-p_3^\star}{2}$.

Denoting
$$\e_3^\star=\sup\limits_{p\in\left[p_3^\star,1\right)}\e_3(p),$$
and using the fact that $B$ is increasing with respect to $\e$, we obtain
\begin{equation}\label{ineq.b23}
B(p,\e_3^\star,n_0)\geq B(p,\e_3(p),n_0)=B_2(p,\e_3(p),n_0)=P_0,\quad\forall~p\in\left(p_3^\star,1\right).
\end{equation}

Let
$$\e^\star=\max\{\e_1^\star,\e_2^\star,\e_3^\star\}.$$
Combining the inequalities (\ref{ineq.b1}), (\ref{ineq.b21}), (\ref{ineq.b22}) and (\ref{ineq.b23}), and taking into consideration that $B$ is increasing with respect to $\e$, we obtain:
$$B(p,\e,n_0)\geq B(p,\e^\star,n_0)\geq P_0,\quad\forall~p\geq\e^\star,~\forall\e\geq\e^\star.$$
and finally:
$$Prob(p,\e,n_0)\geq P_0,\quad\forall~p\geq\e^\star,~\forall\e\geq\e^\star.$$
The proof is now complete.
\end{proof}

\begin{remark}
Instead of the Chernoff-Hoeffding inequalities used in the proof of the previous Theorem, it is possible to compute a lower bound of the correctness probability by means of the well known one-sided Chebyshev inequalities. For example, in a similar manner as in the previous proof, one may obtain the following formula for the estimate of $n^\star$, using Chebyshev bounds:
\begin{equation}\label{n.chebyshev}
n^{\star\star}=\left\{
                 \begin{array}{ll}
                   \ds\left\lceil\max\left\{\frac{P_0(1+\e_0^2)+\sqrt{(1-\e_0^2)^2+4P_0^2\e_0^2}}{8\e_0^2(1-P_0)}~~,~~
                    \frac{P_0}{1-P_0}\cdot\frac{(1-2\e_0)(1+\e_0)^2}{2\e_0(1-\e_0)^2}\right\}\right\rceil, & \textrm{if }\ds\e_0<\frac{1}{3}, \\\\
                   \ds\left\lceil\frac{P_0}{1-P_0}\cdot\frac{1-\e_0}{\e_0}\right\rceil, & \textrm{if }\ds\e_0\geq\frac{1}{3}.
                 \end{array}
               \right.
\end{equation}
However, we note that Chernoff-Hoeffding inequalities lead to better results, since they give exponentially decreasing bounds on tail distributions, while Chebyshev inequalities yield only power-law bounds on tail decay.
\end{remark}

\section{Algorithms for DEPAS tuning}
\label{sec:alg}

In this section, based on the results of the previous section, we formulate the algorithms for estimating the minimum number of nodes and the minimum load variation threshold, respectively, for which the Chernoff-Hoeffding lower bounds of the correctness probability are higher than or equal to the given $P_0$. In order to get an idea of how close the estimations provided by the Chernoff-Hoeffding bounds are to the real minimums, we also provide algorithms for computing \textit{Min n} and \textit{Min $\delta$} based on the binomial formula. The functions used in the algorithms together with their properties are shown in Table \ref{table:functions}. These properties (continuity, differentiability, monotony) determine whether the functions can be used with some numerical optimization algorithms and their specific meaning was explained for each function in the proof of Theorem \ref{th:minimum}.

\begin{table}
    \caption{Functions used in the tuning algorithms}
\begin{center}
\begin{tabularx}{\textwidth}{ | p{6cm} | X | }
	\hline
	\textbf{Function} & \textbf{Properties} \\
	\hline
	$$B_1(p, \e, n) = 1-e^{-nD\left[\frac{p-\e}{1+\e},p\right]}-e^{-nD\left[\frac{p+\e}{1-\e},p\right]}$$ &Continuously differentiable, strictly increasing with respect to $\e$ and $n$ when $p\in[\e,1-2\e],~\e\in\left(0,\frac{1}{3}\right),~n\in\mathbb{Z}_+$  \\
	\hline
	$$B_2(p,\e,n) = 1-e^{-nD\left[\frac{p-\e}{1+\e},p\right]}$$ & Continuously differentiable, strictly increasing with respect to $\e$ and $n$ when $p\geq\max\{\e,1-2\e\},~\e\in(0,1),~n\in\mathbb{Z}_+$ \\
	\hline
	$$s(p) = \frac{-\ln(1-P_0)}{D\left[\frac{p-\e_0}{1+\e_0},p\right]}$$ & Continuously differentiable and concave when $p\in\left(1-2\e_0, 1\right), \e_0\in\left(0, \frac{1}{3}\right)$ \\
	\hline
	$$h_1(p) = 1-(1-p)^{n_0}-e^{-n_0D\left[\frac{2p}{1-p},p\right]}$$ & Continuously differentiable, strictly increasing on $\left(0, \frac{1}{3}\right)$ \\
	\hline
	$$h_2(p) = 1-e^{-n_0D\left[\frac{3p-1}{3-p},p\right]}-p^{n_0}$$ & Continuously differentiable, strictly decreasing on $\left(\frac{1}{3}, 1\right)$ \\
	\hline
	$$h_3(p) = 1-e^{-n_0D\left[\frac{3p-1}{3-p},p\right]}$$ & Continuously differentiable, strictly decreasing on $\left(\frac{1}{3}, 1\right)$\\
	\hline
	$$g(n) = 1-\left(\frac{2}{3}\right)^{n}-\left(\frac{1}{3}\right)^{n}$$ & Continuously differentiable, strictly increasing on $\left(0, \infty\right)$ \\
	\hline
	$$Bin(L, \delta, n) = \sum_{i=\left\lceil M(n, L, L_0 + \delta)\right\rceil}^{\left\lfloor M(n, L, L_0 - \delta) \right\rfloor} \binom{i}{n}p^i(1-p)^{n-i}$$ & Discontinuous and non-monotonic with respect to $L$ and non-monotonic with respect to $\delta$ and $n$ \\
	\hline
	\end{tabularx}
\end{center}
    \label{table:functions}
\end{table}

\subsection{Algorithms based on the Chernoff-Hoeffding bound}

Theorem \ref{th:minimum} used Chernoff-Hoeffding inequalities to derive a lower bound on the correctness probability. The bound is expressed by the $B_1(p, \e, n)$ and $B_2(p, \e, n)$ functions (see Table \ref{table:functions}). By fixing either $\e$ (in the \textit{Min n} scenario) or $n$ (in the \text{Min $\delta$} scenario), $B_1$ and $B_2$ become functions of two variables and one problem is to find the minimum value of the second variable (either $n$ or $\e$, respectively) so that the bound is higher than or equal to the given $P_0$, for any possible value of $p$.

\begin{algorithm}
\caption{Min y monotonic}
\label{alg:min-y-monotonic}
\begin{algorithmic}
\STATE \textbf{Input:} Function $f : I_1 \times I_2 \rightarrow \mathbb{R}$ continuously differentiable on the intervals $I_1$ and $I_2$ and strictly increasing on $I_2$, precisions $s_1$ and $s_2$
\STATE \textbf{Output:} Minimum $y^\star \in I_1$ so that $f(x, y^\star) \ge 0 \forall x \in I_1 $
\STATE Let $l$ be an empty list
\FORALL{$x' \in I_1$ increasing with step $s_1$}
	\STATE Use a root finding method to derive the root of the equation $f(x', y) = 0$ with precision $s_2$.
	\STATE Let $y'$ be the found root.
	\STATE Append $y'$ to $l$
\ENDFOR
\STATE $y^\star \gets max(l)$
\end{algorithmic}
\end{algorithm}

Algorithm \ref{alg:min-y-monotonic} provides a solution to the above problem for a generic function of two variables $f : I_1 \times I_2 \rightarrow \mathbb{R}$, that shares the same properties as $B_1$ and $B_2$. As in general, $B_1$ and $B_2$ are neither monotonic nor convex with respect to $p$, the algorithm has to take increasing numbers from interval $I_2$ with precision $s_1$. Then, taking advantage of the fact that $B_1$ and $B_2$ are strictly increasing with respect to $\e$ and $n$, for each considered value of $y' \in I_2$ we can use a root finding method to compute the minimum $x'$ for which $f(x', y') \ge 0$. Then, the maximum of the $x'$ values satisfies the inequality for all $y'$ values that were considered.

\begin{algorithm}
\caption{Cernoff min n}
\label{alg:cernoffMinN}
\begin{algorithmic}
\STATE \textbf{Input:} $L_0 \in (0,1), \delta_0 \in (0, L_0), P_0 \in (0, 1)$, precisions $s_n$ and $s_p$
\STATE \textbf{Output:} $n^\star \in \mathbb{Z}_+$
\STATE Compute $N$ with formula \ref{n.chebyshev}
\STATE  $\e_0 \gets \frac{\delta_0}{L_0}$
\IF {$\e_0 \in \left(0, \frac{1}{3}\right)$}
	\STATE Use algorithm \ref{alg:min-y-monotonic} to find the minimum $n_1^\star \in 1..N$ for which $f(p, n_1^\star) = B_1(p, \e_0, n_1^\star) - P_0 \ge 0, \forall p\in\left[\e_0,1-2\e_0\right], s_1 = s_p, s_2 = s_n$
	\STATE Use a concave optimization method to find $n_2^\star \gets \sup\limits_{p\in(1-2\e_0,1)}s(p)$ with precision $s_p$
	\STATE $n^\star=\left\lceil\max\{n_1^\star,n_2^\star\}\right\rceil$
\ENDIF
\IF {$\e_0 \in \left[\frac{1}{3}, 1\right)$}
	\STATE Use a concave optimization method to find $n_3^\star \gets \sup\limits_{p\in[\e_0,1)}s(p)$ with precision $s_p$
	\STATE $n^\star=\left\lceil n_3^\star\right\rceil$
\ENDIF
\end{algorithmic}
\end{algorithm}

Algorithm \ref{alg:min-y-monotonic} is needed by both Algorithm \ref{alg:cernoffMinN} (i.e., \textit{Cernoff min n}) and Algorithm \ref{alg:cernoffMinDelta} (i.e., \textit{Cernoff min $\delta$}). In fact, Algorithms \ref{alg:cernoffMinN} and \ref{alg:cernoffMinDelta} just translate the proof of Theorem \ref{th:minimum} into algorithmic language. Note that Algorithm \ref{alg:cernoffMinN} uses the Chebyshev inequality (which is weaker than the Chernoff-Hoeffding inequality) to compute an upper bound of the number of nodes.

\begin{algorithm}
\caption{Cernoff min $\delta$}
\label{alg:cernoffMinDelta}
\begin{algorithmic}
\STATE \textbf{Input:} $L_0 \in (0,1), n_0 \in \mathbb{Z}_+, P_0 \in (0,1)$, precisions $s_{\e}$ and $s_p$
\STATE \textbf{Output:} $\delta^\star \in (0, L_0)$
\IF {$1-\left(\frac{2}{3}\right)^{n_0}-\left(\frac{1}{3}\right)^{n_0} \le P_0$}
	\STATE \textbf{exit} with no solution
\ELSE	
	\STATE Use a root finding method to derive the unique solution $p_1^\star \in \left(0, \frac{1}{3}\right)$ of the equation $h1(p) = P_0$ with precision $s_p$
	\STATE Use algorithm \ref{alg:min-y-monotonic} to find the minimum $\e_1^\star < p_1^\star$ for which $f(p, \e_1^\star) = B_1(p, \e_1^\star, n_0) - P_0 \ge 0, \forall p\in\left(p_1^\star, \frac{1}{3}\right], s_1 = s_p, s_2 = s_{\e}$
	\STATE Use a root finding method to derive the unique solution $p_2^\star \in \left(\frac{1}{3}, 1\right)$ of the equation $h2(p) = P_0$ with precision $s_p$
	\STATE Use algorithm \ref{alg:min-y-monotonic} to find the minimum $\e_2^\star < p_2^\star$ for which $f(p, \e_2^\star) = B_1(p, \e_2^\star, n_0) - P_0 \ge 0, \forall p\in\left(\frac{1}{3}, p_2^\star\right], s_1 = s_p, s_2 = s_{\e}$
	\STATE Use a root finding method to derive the unique solution $p_3^\star \in \left(\frac{1}{3}, 1\right)$ of the equation $h3(p) = P_0$ with precision $s_p$
	\STATE Use algorithm \ref{alg:min-y-monotonic} to find the minimum $\e_3^\star < p_3^\star$ for which $f(p, \e_3^\star) = B_2(p,\e_3^\star, n_0) - P_0 \ge 0, \forall p \in\left[p_3^\star, 1\right), s_1 = s_p, s_2 = s_{\e}$
	\STATE $\e^\star \gets \max\left\{\e_1^\star,\e_2^\star,\e_3^\star\right\}$
	\STATE $\delta^\star \gets \e^\star L_0$
\ENDIF
\end{algorithmic}
\end{algorithm}

Note that, from a theoretical point of view, the accuracy of Algorithm \ref{alg:min-y-monotonic} (and implicitly that of Algorithms \ref{alg:cernoffMinN} and \ref{alg:cernoffMinDelta}) may be influenced by the precision $s_1$. The possible correlation between the output accuracy and the precision is checked experimentally in Section \ref{sec:tests}.

\subsection{Algorithms based on the binomial formula}
The binomial formula, expressed by function $Bin(L, \delta, n)$ from Table \ref{table:functions}, allows us to compute the exact value of the correctness probability. However, it is more difficult to compute the minimum $n$ or $\delta$ based on the binomial formula than based on the Chernoff-Hoeffding bounds because the binomial function $Bin(L, \delta, n)$ is discontinuous with respect to $L$ and non-monotonic with respect to $\delta$ and $n$. A solution is given by Algorithm \ref{alg:min-y}, where $x$ accounts for either $\delta$ or $n$ and $y$ represents $L$. The main idea is to check for each $y$ starting from its maximum possible value and decreasing by a given step whether the function is positive for all $x$ taken with a certain precision. The solution is given by the minimum $y$ for which the function is positive for all considered values of $x$.

\begin{algorithm}
\caption{Min y}
\label{alg:min-y}
\begin{algorithmic}[1]
\STATE \textbf{Input:} Function $f : I_1 \times I_2 \rightarrow \mathbb{R}$, precisions $s_1$ and $s_2$
\STATE \textbf{Output:} Minimum $y^\star \in I_1$ so that $f(x, y^\star) \ge 0 \forall x \in I_1 $
\STATE $y' \gets 0$
\FORALL{$y \in I_2$ decreasing with step $s_2$}
	\STATE $y' \gets y$
	\FORALL{$x \in I_1$ increasing with step $s_1$}
		\IF{$f(x,y) < 0$}
			\STATE \textbf{go to} \ref{marker}
		\ENDIF
	\ENDFOR
\ENDFOR
\STATE $y^\star \gets y' + s_2$ \label{marker}
\end{algorithmic}
\end{algorithm}

Concrete algorithms for approximating the minimum $n$ (\textit{Binomial min n}) and the minimum $\delta$ (\textit{Binomial min $\delta$}) are obtained by instantiating Algorithm \ref{alg:min-y} for functions $f_1(x,y) = f_1(L, n) = Bin(L,\delta_0,n) - P_0$ and $f_2(x,y) = f_2(L,\delta) = Bin(L,\delta,n_0) - P_0$, respectively. A maximum bound for $n$ can be computed based on either Chebyshev or Chernoff-Hoeffding inequalities.

Note that, due to the discontinuity of the binomial function, the binomial algorithms are prone to numerical errors and, therefore, not recommended to be used by the customer of DEPAS. Their sole objective is to allow us to estimate the accuracy of the Chernoff algorithms as shown in the next section.

\section{Experimental results}
\label{sec:tests}

The DEPAS tuning algorithms described in the previous section were implemented in Java. We used simple root and supremum finding algorithms based on the bisection method. The DEPAS Tuning Tool, which provides both graphical and command line facilities for running the tuning algorithms in a large spectrum of scenarios, is available for download \cite{depas-tuning-tool}.

The DEPAS Tuning Tool was used to obtain the experimental results presented in this section. There are three objectives of the experiments: (i) to estimate the accuracy of Chernoff algorithms by comparing their results with the ones of the binomial algorithms, (ii) to make an idea about the execution time of the tuning algorithms, and (iii) to see how the accuracy of Chernoff algorithms is affected by the precision of the load. Two experiments (one with \textit{Min n} algorithms and another one with \textit{Min $\delta$} algorithms) were performed for each objective, thus resulting in a total of six experiments.

All experiments were performed on Amazon EC2 Large Instances (7.5 GB memory, 4 EC2 Compute Units, 64-bit platform) running Amazon Linux. In all experiments we set $L_0 = 0.8$, and $P_0 = 0.99$. The default precisions were $s_n = 10^{-1}, s_{\e} = 10^{-3}$, and $s_p = 10^{-3}$.

\begin{figure}
    \begin{center}
        \includegraphics[scale=1]{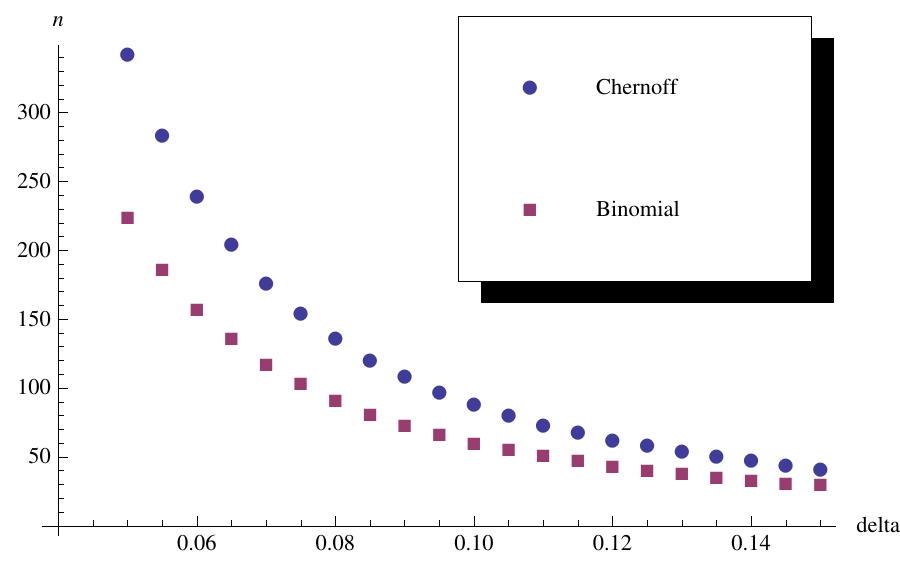}
        \caption{\textit{Min n} experimental results: Chernoff vs. Binomial}
        \label{fig:min-n-results}
    \end{center}
\end{figure}

In the first experiment, the number of nodes was computed with both Chernoff and binomial algorithms for several values of $\delta$ taken from the interval $\left[0.05, 0.15\right]$ with step $0.005$. From Figure \ref{fig:min-n-results}, we can see that \textit{Chernoff min n} and \textit{Binomial min n} algorithms give relatively close results and that the results become closer with the increase of $\delta$. For example, for $\delta = 0.05$ Chernoff gives $n = 342$ and binomial gives $n = 224$, while for  $\delta = 0.15$ the result of Chernoff is 41 compared to binomial's 30.

In the second experiment the values of $\delta$ were computed with both Chernoff and binomial algorithms for several values of $n$ taken from the interval $\left[25, 1000\right]$ with step $5$. As shown in Figure \ref{fig:min-n-results}, the results are close and become closer with higher values of $n$. To give just two examples, for $n = 100$, Chernoff gives $\delta = 0.094$ and binomial gives $\delta = 0.075$, while for $n = 1000$ Chernoff's $\delta$ is $0.03$ and binomial's is $0.023$.

\begin{figure}
    \begin{center}
        \includegraphics[scale=1]{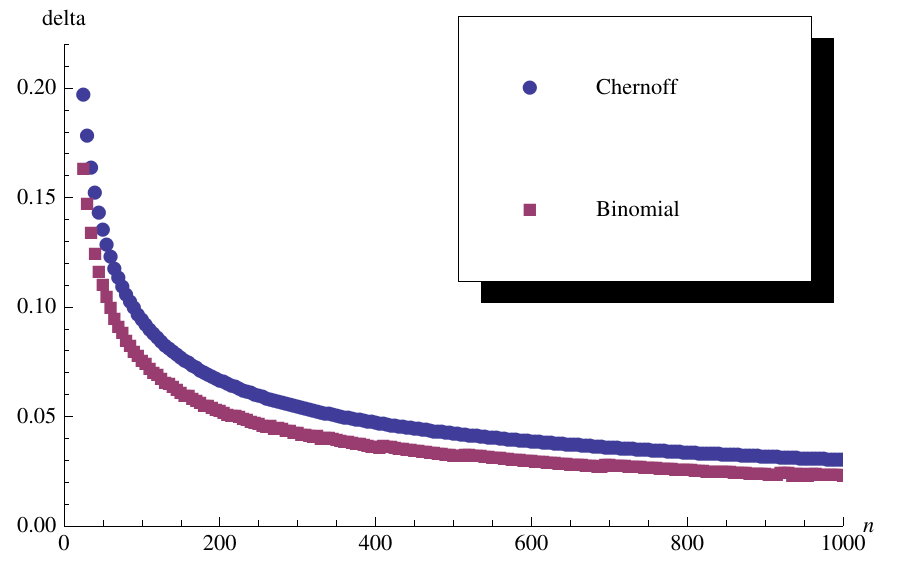}
        \caption{\textit{Min $\delta$} experimental results: Chernoff vs. Binomial}
        \label{fig:min-n-results}
    \end{center}
\end{figure}

The third experiment was actually a series of 32 identical, but independently performed experiments that ran the \textit{Cernoff min n} and \textit{Binomial min n} algorithms for several values of $\delta$. The average execution time of both algorithms is shown in Table \ref{table:execution-time}. Note that the result of the \textit{Cernoff min n} algorithm was used as an upper bound in the \textit{Binomial min n} algorithm and, therefore, the execution time of the latter includes the execution time of the former. We can notice that \textit{Cernoff min n} is very fast no matter what the value of $\delta$ is, while \textit{Binomial min n} is slow for small values of $\delta$, but its execution time decreases with the increase of $\delta$.

The fourth experiment derived the execution time of the \textit{Cernoff min $\delta$} and \textit{Binomial min $\delta$} algorithms as an average of the execution times obtained from a set of 32 identical experiments. Several values of $n$ were considered and the results are shown in Table \ref{table:execution-time}. We can see that Chernoff performs extremely fast again, while Binomial is slow for high values of $n$, improving when $n$ decreases.

\begin{table}
\begin{center}
    \caption{Average execution time (in seconds) for \textit{Min n} and \textit{Min $\delta$}: Chernoff vs. Binomial}
    \label{table:execution-time}
\begin{tabular}{ | r | r | r || r | r | r |}
	\hline
	\multicolumn{3}{|c||}{\textit{Min n}} & \multicolumn{3}{|c|}{\textit{Min $\delta$}} \\ \hline
	$\delta$ & Chernoff & Binomial & n & Chernoff & Binomial \\ \hline
	0.05 &0.013 & 280.804 & 25 & 0.011 & 1.580 \\
	0.075 &0.010 & 29.461 & 250 & 0.021 & 29.004\\
	0.1 & 0.008 & 6.496 & 500 & 0.020 & 89.767\\
	0.125 & 0.006 & 2.081 & 750 & 0.014 & 178.924\\
	0.15 & 0.005 & 0.803 & 1000 & 0.016 & 310.619\\
	\hline
\end{tabular}
\end{center}
\end{table}

The last two experiments checked whether the results of \textit{Chernoff min n} and \textit{Chernoff min $\delta$} algorithms are influenced by the precision of the load, $s_p$. Therefore, experiments 5 and 6 executed the Chernoff algorithms in the same conditions as experiments 1 and 2, respectively, but for three different load precisions: $10^{-3}, 10^{-4}, 10^{-5}$. We found that the three-decimal results computed for the above load precisions were identical in all the considered cases, which means that it is worthless to use lower than $10^{-3}$ load precisions.

In conclusion, the Chernoff algorithms produce results that are close to the real minimums, are very fast, and their results accuracy does not improve with low load precisions. Therefore, they can be used in both \textit{Min n} and \textit{Min $\delta$} scenarios at both design time and runtime.

\section{Related work}
\label{sec:related}

Randomized algorithms have found widespread applicability due to their simplicity and speed \cite{motwani1996}. A brief survey on randomized algorithms \cite{motwani1996} as well as a more recent and detailed one \cite{motwani2010} were written by Motwani and Raghavan. The probabilistic analysis of randomized algorithms aims to provide probabilistic guarantees with respect to the likelihood of these algorithms to perform correctly or efficiently. For example, in the framework of resource allocation and admission control in transactional systems, Almeida et al. express the quality of service as the probability of executing every job within a maximum timeframe, and use the Markov and Chebyshev inequalities to compute an upper bound of this probability \cite{almeida2010}.

Particularly, randomization proved to be a powerful tool for building decentralized algorithms that run in parallel and work with local information. Thus, probabilistic decentralized algorithms were proposed for solving the load balancing problem in large networks \cite{mitzenmacher2001, fu2011}.  Theoretical analysis proved that, by using these algorithms, the expected or the maximum response time is less than a certain threshold with high probability. As for the Chernoff bound, it was applied in the probabilistic analysis of randomized packet routing algorithms for sparse networks and helped proving that certain algorithms are able to route all packets in a maximum number of steps with high probability \cite{mitzenmacher2005, upfal1984}.

A recent application of the Chernoff bound can be found in the decentralized and probabilistic solution to document clustering that was proposed by Papapetrou et al. \cite{papapetrou2011}. In their approach, the correctness probability is the probability of assigning each document to the right cluster and -- similarly to our case -- is set by the user and used to tune the parameters of the algorithm. 

\section{Conclusions}
\label{sec:conclusion}

The decentralized probabilistic auto-scaling (DEPAS) algorithm can be used to deploy elastic service systems that can quickly grow from tens to tens of thousands of computing nodes. However, the potential customer of DEPAS may not be convinced only by the experimental results that are, after all, dependent on the specific experimental scenarios being considered.

Therefore, in this paper, we defined the correctness of DEPAS in a probabilistic manner and modeled DEPAS as a set of Bernoulli trials. Then, the Chernoff-Hoeffding bounds were used to theoretically prove that there is a subset of configurations for which a minimum correctness probability can be guaranteed. Moreover, based on the theoretical results, we designed numerical algorithms for automatically tuning DEPAS so that it can be both correct and accurate. Through a set of experiments we showed that the results given by the Chernoff-based tuning algorithms are closed to the real minimums, which were estimated based on the binomial formula. 

In a future paper, the \textit{Chernoff min $\delta$} algorithm will be used at runtime to dynamically adapt DEPAS, and the expected gain in accuracy will be experimentally checked for a set of realistic workload traces. 

\section{Acknowledgements}
This research has been partially funded by the Romanian National Authority for Scientific Research, CNCS Ð UEFISCDI, under project PN-II-ID-PCE-2011-3-0260 (AMICAS) and by the European Commission, under project FP7-ICT-2009-5-256910 (mOSAIC). Bogdan Caprarescu is partially supported by IBM through a PhD Fellowship Award.






\bibliographystyle{model1-num-names}
\bibliography{jcss}







\end{document}